\documentclass[11pt]{article}
\usepackage[left=1in,top=1in,right=1in,bottom=1in,head=.1in,nofoot]{geometry}

\usepackage{setspace,url,bm,amsmath} 

\usepackage[numbers]{natbib}
\usepackage{graphicx} 
\usepackage{bbm}
\usepackage{latexsym}
\usepackage{caption}
\usepackage[margin=20pt]{subcaption}
\usepackage{hyperref}
\usepackage[]{algorithm2e}

\usepackage[table]{xcolor}

\newcommand{\GG}[1]{}

\usepackage{amsthm}
\usepackage{amssymb}
\usepackage{amsmath}
\usepackage{color}

\usepackage{comment}
\theoremstyle{definition}

\newtheorem*{theorem*}{Theorem}
\newtheorem{theorem}{Theorem}

\newtheorem{lemma}{Lemma}

\newtheorem*{corollary*}{Corollary}

\usepackage{natbib} 
\bibpunct{(}{)}{;}{a}{}{,} 

\usepackage{color}
\usepackage{listings}

\def\sumn{\sum_{j=1}^n}

\def\Hyper{\text{HyperGeo}}
\def\Tau{\bm{\mathcal{T}}}
\def\obs{\text{obs}}

\begin{document}
\doublespacing
\title{\bf Exact confidence intervals for the average causal effect on a binary outcome}
\author{Xinran Li~\footnote{Ph.D. Candidate, Department of Statistics, Harvard University}
~and Peng Ding~\footnote{Postdoctoral Research Fellow, Department of Epidemiology, Harvard T. H. Chan School of Public Health}~\footnote{Corresponding author: Peng Ding, Email: \texttt{pengdingpku@gmail.com}}
}
\date{}
\maketitle

\begin{abstract}
Based on the physical randomization of completely randomized experiments, \citet*{Rigdon:2015} propose two approaches to obtaining exact confidence intervals for the average causal effect on a binary outcome. They construct the first confidence interval by combining, with the Bonferroni adjustment, the prediction sets for treatment effects among treatment and control groups, and the second one by inverting a series of randomization tests. With sample size $n$, their second approach requires performing $O(n^4)$ randomization tests. We demonstrate that the physical randomization also justifies other ways to constructing exact confidence intervals that are more computationally efficient. By exploiting recent advances in hypergeometric confidence intervals and the stochastic order information of randomization tests, we propose approaches that either do not need to invoke Monte Carlo, or require performing at most $O(n^2)$ randomization tests. We provide technical details and R code in the Supplementary Material. 

\noindent {\bf Keywords}: Causal inference; completely randomized experiment; potential outcome; randomization test; two by two table
\end{abstract}

\section{Notation and framework}

We extend the notation in \citet{Rigdon:2015}. In a completely randomized experiment with $n$ units, let $Z_j$ and $Y_j$ denote the binary treatment assignment and the binary outcome for unit $j$. We define $y_j(1)$ and $y_j(0)$ as the potential outcomes of unit $j$ under treatment and control, and let $N_{ik} = \#\{j: y_j(1)=i, y_j(0)=k\}$ for $i, k=0,1$. Following \citet{Ding:2015}, the potential table $\boldsymbol{N}=(N_{11}, N_{10}, N_{01}, N_{00})$  summarizes the potential outcomes for all units. The total numbers of units with treatment and control potential outcomes being one are $N_{1+}=\sumn y_j(1)$ and $N_{+1}=\sumn y_j(0)$, respectively. The individual causal effect for unit $j$ is $\delta_j=y_j(1)-y_j(0)$, and the average causal effect is $\tau(\boldsymbol{N}) = \sumn \delta_j/n = (N_{1+}-N_{+1})/n =  (N_{10}-N_{01})/n$. Here we emphasize that $\tau(\boldsymbol{N})$ is a function of $\boldsymbol{N}$, and later we write it as $\tau$ for simplicity. Let $\boldsymbol{Z}=(Z_1, Z_2, \ldots, Z_n)$ be the treatment assignment vector, and the treated group is a simple random sample of size $m$ from the $n$ experimental units. The observed outcome of unit $j$ is $Y_j=Z_jy_j(1)+(1-Z_j)y_j(0)$, a deterministic function of $Z_j$ and the potential outcomes $(y_j(1), y_j(0))$. We can summarize the observed data by four counts $n_{zy} = \#\{ j: Z_j = z, Y_j = y  \}$ for $z,y = 0,1$, and call $\boldsymbol{n} = (n_{11}, n_{10}, n_{01}, n_{00})$ the observed table. The intuitive estimator, $\hat{\tau} = n_{11} /m - n_{01} /(n-m)$, is unbiased for $\tau.$


Before observing the data, the potential table $\boldsymbol{N}$ can take any values as long as the sum of the $N_{ik}$'s is $n$. After obtaining $\boldsymbol{n}$, the data put some restrictions on the potential table. 
A potential table $\boldsymbol{N}$ is compatible with the observed table $\boldsymbol{n}$, if there exist potential outcomes $\{  (y_j(1), y_j(0)) \}_{j=1}^n$, summarized by $\boldsymbol{N}$, that give the observed table $\boldsymbol{n}$ under the treatment assignment $\boldsymbol{Z}$.


\begin{theorem}\label{compat_st}
A potential table $\boldsymbol{N} $ is compatible with the observed table $\boldsymbol{n} $  if and only if
\begin{align*}
\max\{0, n_{11}-N_{10}, N_{11}-n_{01}, N_{+1} -n_{10}-n_{01}\} \leq 
\min\{N_{11}, n_{11}, N_{+1}-n_{01}, n-N_{10}-n_{01}-n_{10}\} .
\end{align*}
\end{theorem}

Theorem \ref{compat_st} gives an easy-to-check condition, which plays an important role in our later discussion. For all potential tables compatible with the observed table, their $\tau$ values must be equal to some $k/n$, with integer $k$ between $-(n_{10}+n_{01})$ and $ n_{11}+n_{00}$ \citep{Rigdon:2015}.

\section{Confidence intervals without Monte Carlo}
\label{sec:confidence_without_MC}
We propose two approaches to constructing confidence intervals for $\tau$ based on the hypergeometric distribution, which avoid Monte Carlo and are easy to compute. Let $X\sim \Hyper(A,T,S)$ denote the hypergeometric distribution representing the number of units having some attribute in a simple random sample of size $S$, which are drawn from $T$ units with $A$ units having this attribute. Recently, \citet{Wang:2015} improves classical hypergeometric confidence intervals, and proposes an optimal procedure to construct a confidence interval for $A$ based on $(T,S,X)$. Our discussion below relies on this confidence interval for $A$ based on a $\Hyper(A,T,S)$ random variable $X$.

\subsection{Combining confidence intervals for $N_{1+}$ and $N_{+1}$}
\label{subsec:combine_two_CI}
We can construct an exact confidence interval for $\tau = (N_{1+}-N_{+1})/n$, by combining confidence intervals for $N_{1+}$ and $N_{+1}$ with the Bonferroni adjustment. Because the treated and control units are simple random samples of the $n$ units in a completely randomized experiment, we have
\begin{align*}
n_{11} \sim \Hyper\left(N_{1+}, n, m\right), \quad
n_{01} \sim \Hyper\left( N_{+1}, n, n-m\right).
\end{align*}
We first obtain $(1-\alpha/2)$ confidence intervals, $[N_{1+}^L, N_{1+}^U]$ and $[N_{+1}^L, N_{+1}^U]$,  for $N_{1+}$ and $N_{+1}$, and then use $[ (N_{1+}^L -N_{+1}^U)/n,  (N_{1+}^U -N_{+1}^L)/n   ]$ as a $(1-\alpha)$ confidence interval for $\tau$.

\subsection{A test statistic with simple null distributions}\label{invert_easy_null}
We can construct exact confidence intervals for $\tau$ by inverting a series of randomization tests. However, the null distributions of $| \hat{\tau}-\tau | $ is complex \citep{Rigdon:2015}.   If we use the difference between the average causal effect on the treated units and $\hat{\tau}$ as the test statistic, then the null distribution has a simple form. The test statistic 
\begin{align}\label{causal_effect_on_treat}
\frac{1}{m}\sumn Z_j\delta_j-\hat{\tau} 
= \frac{n}{m(n-m)}n_{01}-\frac{1}{m}N_{+1}
\end{align}
is equivalent to $n_{01} \sim \Hyper(N_{+1}, n, n-m)$, because \eqref{causal_effect_on_treat} is a monotone function of $n_{01}$, its only random component.
Because the null distribution of $n_{01}$ depends only on $N_{+1}$, potential tables with the same value of $N_{+1}$ will yield the same $p$-value under randomization tests. Therefore, we need only to perform $O(n)$ randomization tests according to all possible values of $N_{+1}$ between $n_{01}$ and $n-n_{00}$. The final lower and upper confidence limits for $\tau$ are the minimum and maximum values of $\tau (\boldsymbol{N})$ subject to (a) $\boldsymbol{N}$ is compatible with $\boldsymbol{n}$, and (b) $\boldsymbol{N}$ yields a $p$-value larger than or equal to $\alpha$. Constraint (b) is equivalent to restricting $N_{+1}$ within a $(1-\alpha)$ confidence interval $[N_{+1}^L, N_{+1}^U]$, which helps avoid randomization tests or Monte Carlo.


\section{Two-sided confidence intervals with fewer randomization tests}
\label{subsec:two_side}
We consider two sided confidence intervals for $\tau$ using $|\hat{\tau}-\tau|$ as the test statistic. We define $p_2(\boldsymbol{N}) = P_{\boldsymbol{N}}(|\hat{\tau}-\tau| \geq |\hat{\tau}^{\text{obs}}-\tau|)$ with $\hat{\tau}^{\text{obs}}$ being the realized value of $\hat{\tau}$, which is the $p$-value of potential table $\boldsymbol{N}$, or equivalently a null hypothesis, that is compatible with the observed table. 
We need to find all potential tables with $p_2 \geq \alpha$, then use Theorem \ref{compat_st} to find the compatible ones among them, and eventually find the maximum and minimum $\tau$ values. Without loss of generality, we assume $m\leq n/2$; otherwise we can switch the labels of $Z$. 
We consider the potential tables with $\tau \leq \hat{\tau}^\obs$ and $\tau \geq \hat{\tau}^\obs$ separately. 
We first focus on the potential tables with $\tau \leq \hat{\tau}^\obs$, and rule out the ones with $p_2<\alpha$.



For given $N_{11}$ and $N_{01}$, we define $\underline{N}_{10}(N_{11},N_{01})$ as the minimum value of $N_{10}$ such that $p_2(\boldsymbol{N}) \geq \alpha$ and $\tau(\boldsymbol{N}) \leq \hat{\tau}^{\text{obs}}$, with $\boldsymbol{N} = (N_{11}, N_{10}, N_{01}, n-N_{11}-N_{10}-N_{01})$ being a potential table. If there is no such $N_{10}$, then we define $\underline{N}_{10}(N_{11},N_{01}) = \lfloor N_{01} + n \hat{\tau}^{\obs} \rfloor+1$, the smallest value of $N_{10}$ such that $\tau(\boldsymbol{N}) > \hat{\tau}^\obs$, where $\lfloor x \rfloor$ is the largest integer less than or equal to $x$.

\begin{theorem}\label{thm::order-two-sided}
\begin{itemize}
\item[(1)] 
If $N_{01} \leq N_{01}'$, then $\underline{N}_{10}(N_{11}, N_{01}) \leq \underline{N}_{10}(N_{11}, N_{01}')$.

\item[(2)] 
In balanced experiments with $m=n/2$, consider a potential table $\boldsymbol{N}$ with $\tau(\boldsymbol{N}) \leq \hat{\tau}^{\text{obs}}$. $p_2(\boldsymbol{N}) \geq \alpha$ if and only if $N_{10} \geq \underline{N}_{10}(N_{11}, N_{01})$.

\end{itemize}
\end{theorem}

The computation burden arises because we need to perform randomization tests for all potential tables compatible with the observed table. Fortunately, Theorem \ref{thm::order-two-sided}(1) provides useful order information to reduce the number of randomization tests. We first assume that $N_{11}$ is fixed. When $N_{01} = 0$, we find $\underline{N}_{10}(N_{11}, 0)$ by performing randomization tests starting from $N_{10}=0$. When $N_{01}$ increases to $1$, we find $\underline{N}_{10}(N_{11}, 1)$ by performing randomization tests starting from $N_{10} = \underline{N}_{10}(N_{11}, 0)$ according to Theorem \ref{thm::order-two-sided}(1). Sequentially, when $N_{01}$ increases by $1$, we find $\underline{N}_{10}(N_{11}, N_{01})$ by performing randomization tests starting from $N_{10} = \underline{N}_{10}(N_{11}, N_{01}-1)$. We repeat this process until $N_{10}$ increases to $n-N_{11}$. For a fixed $N_{11}$, we need to perform at most $O(n)$ randomization tests. We implement the above procedure with $N_{11}$ increasing from $0$ to $n_{11}+n_{01}$, which requires at most $O(n^2)$ randomization tests in total.  

As long as we find $\underline{N}_{10}(N_{11}, N_{01})$ for all possible $(N_{11}, N_{01})$, 
we accept the potential tables compatible with the observed table such that $N_{10}\geq \underline{N}_{10}(N_{11}, N_{01})$ and $\tau(\boldsymbol{N})\leq \hat{\tau}^\obs$, with $\boldsymbol{N} = (N_{11},N_{10},N_{01},n-N_{11}-N_{10}-N_{01})$. By switching the labels of $Y$, we can similarly accept some potential tables with $\tau\geq \hat{\tau}^\obs$. The final lower and upper confidence limits for $\tau$ are the minimum and maximum $\tau$ values of accepted potential tables.


The preceding confidence interval might be slightly wider than the interval obtained by the second approach in \citet*{Rigdon:2015}, because potential tables with $N_{10}\geq \underline{N}_{10}(N_{11}, N_{01})$, $\tau(\boldsymbol{N})\leq \hat{\tau}^\obs$ might not satisfy $p_2\geq \alpha$. However, this interval will be the same as \citet*{Rigdon:2015} in balanced experiments according to Theorem \ref{thm::order-two-sided}(2). Our numerical and extensive simulation studies demonstrate that the above confidence interval coincides with the second one in \citet*{Rigdon:2015} even though the experiments are extremely unbalanced.

In balanced experiments, for every $\tau$ within the above confidence interval, there exists a potential table $\boldsymbol{N}$ compatible with the observed table such that $\tau(\boldsymbol{N})=\tau$ and $p_2(\boldsymbol{N})\geq \alpha$. For general experiments, we can use the same idea to construct an exact one-sided confidence interval with $O(n^2)$ randomization tests. We comment on these two issues in the Supplementary Material.

\section{Numerical examples}
We compare the confidence intervals for $n\tau$ obtained by various procedures. Table \ref{table:numerical_result} shows the observed tables and results. In our examples, the confidence intervals in Section \ref{subsec:two_side} are the same as the second approach in \citet*{Rigdon:2015}, even though some of them are very unbalanced. We have conducted extensive simulations for all observed tables with $n=24$, finding that these two methods give the same $90\%, 95\%$ and $99\%$ confidence intervals.

\begin{table}[ht]
\caption{$95\%$ confidence intervals for $n\tau$. ``2.1'', ``2.2'' and ``3'' denote the methods in Sections \ref{subsec:combine_two_CI}, \ref{invert_easy_null} and \ref{subsec:two_side}; I and II denote the first and second approaches in \citet*{Rigdon:2015}; ``\#(3)'' and ``\#(II)'' denote the numbers of randomization tests needed for ``3'' and ``II.''
}
\label{table:numerical_result}
\begin{center}
\begin{tabular}{|c|c|c|c|c|c|c|c|}
\hline
$\boldsymbol{n}$ & 2.1 & 2.2 & 3 & I & II & \#(3) & \#(II) \\
\hline
$(1, 1, 1, 13)$ & $[-2, 14]$ & $[-1,14]$  & $[-1,14]$ & $[-2,14]$ & $[-1,14]$ & 103 & 112\\
$(2,6, 8, 0)$ & $[-14, -3]$ & $[-14,-2]$ & $[-14,-5]$ & $[-14,-2]$ & $[-14,-5]$ & 113 & 189\\
$(6,0, 11, 3)$ & $[-5, 8]$ & $[-11,7]$  & $[-4,8]$ & $[-5,7]$ & $[-4,8]$ & 283 & 336\\
$(6, 4, 4, 6)$ & $[-6, 12]$ & $[-6, 11]$  & $[-4, 10]$ & $[-6, 12]$ & $[-4, 10]$ & 308 & 1225\\
$(1,  1,  3, 19)$ & $[-4, 20]$ & $[-3, 20]$  & $[-3,   20]$ & $[-4, 20]$ & $[-3, 20]$ & 251 & 320 \\
$(8,4,5,7)$ & $[-4, 14]$ & $[-6, 14]$ &  $[-3, 13]$ & $[-6, 15]$  & $[-3, 13]$ & 421 & 2160\\
\hline
\end{tabular}
\end{center}
\end{table}

\bibliographystyle{plainnat}
\bibliography{causal}

\newpage
\setcounter{page}{1}
\begin{center}
\bf \huge 
Supplementary material
\end{center}

\bigskip

\setcounter{equation}{0}
\setcounter{section}{0}
\setcounter{figure}{0}
\setcounter{example}{0}
\setcounter{proposition}{0}
\setcounter{theorem}{0}
\setcounter{table}{0}

\renewcommand {\theproposition} {A.\arabic{proposition}}
\renewcommand {\theexample} {A.\arabic{example}}
\renewcommand {\thefigure} {A.\arabic{figure}}
\renewcommand {\thetable} {A.\arabic{table}}
\renewcommand {\theequation} {A.\arabic{equation}}
\renewcommand {\thelemma} {A.\arabic{lemma}}
\renewcommand {\thesection} {A.\arabic{section}}
\renewcommand {\thetheorem} {A.\arabic{theorem}}

Section \ref{subsec:one_side} describes a procedure to obtain a one-sided confidence interval using $\hat{\tau}$ as the test statistic, which requires $O(n^2)$ randomization tests. Section \ref{sec:proof_for_theorems} contains proofs of the theorems. Section \ref{sec:more_computation_detail} comments on some computational details. In the following, we use $a \vee b = \max(a,b)$ and $ a\wedge b = \min(a,b).$

\section{One-sided confidence interval}
\label{subsec:one_side}

Without loss of generality, we consider the lower confidence limit for $\tau$ using $\hat{\tau}$ as the test statistic in randomization tests. We define 
$
p_1(\boldsymbol{N}) = P_{\boldsymbol{N}}(\hat{\tau}\geq \hat{\tau}^{\text{obs}})
$  
with $\hat{\tau}^{\text{obs}}$ being the realized value of $\hat{\tau}$, which is the $p$-value for potential table $\boldsymbol{N}$ compatible with observed table. To obtain the $(1-\alpha)$ lower confidence limit, we need to find all potential tables with $p_1 \geq \alpha$, then use Theorem \ref{compat_st} to find the compatible ones among them, and eventually find the minimum $\tau$ value.

To facilitate computation, it is crucial to exploit order information of the potential tables. For given $  N_{11} $ and $  N_{01} $,  we define $\underline{N}_{10}(N_{11}, N_{01})$ as the minimum value of $N_{10}$ such that $p_1(\boldsymbol{N}) \geq \alpha $ with $\boldsymbol{N} = (N_{11},N_{10},N_{01},n-N_{11}-N_{10}-N_{01})$ being a potential table. If there is no such $N_{10}$, then we define $\underline{N}_{10}(N_{11}, N_{01}) = n+1$.

\begin{theorem}
\label{thm::order-one-sided}
\begin{itemize}
\item[(1)] 
If $N_{01}\leq N'_{01}$, then $\underline{N}_{10}(N_{11}, N_{01}) \leq \underline{N}_{10}(N_{11},N'_{01})$.

\item[(2)] 
For a potential table $\boldsymbol{N}$, $p_1(\boldsymbol{N}) \geq \alpha$ if and only if $N_{10} \geq \underline{N}_{10}(N_{11}, N_{01})$.
\end{itemize}
\end{theorem}

Theorem \ref{thm::order-one-sided} provides useful order information to reduce the number of randomization tests, because we do not need to test for all potential tables. According to Theorem \ref{thm::order-one-sided}(2), we need only to find $\underline{N}_{10}(N_{11}, N_{01})$ for given $  (N_{11}, N_{01}) $.
We first fix $N_{11}$. When $N_{01} = 0$, we find $\underline{N}_{10}(N_{11}, 0)$ by performing randomization tests starting from $N_{10}=0$. When $N_{01}$ increases to $1$, we find $\underline{N}_{10}(N_{11}, 1)$ by performing randomization tests starting from $N_{10} = \underline{N}_{10}(N_{11}, 0)$, because Theorem \ref{thm::order-one-sided}(1) guarantees that $\underline{N}_{10}(N_{11}, 1) \geq \underline{N}_{10}(N_{11}, 0).$ We repeat this process until $N_{10}$ increases to $n-N_{11}$. For a fixed $N_{11}$, we need to perform at most $O(n)$ randomization tests. We implement the above procedure with $N_{11}$ increasing from $0$ to $n_{11}+n_{01}$, which requires at most $O(n^2)$ randomization tests in total. For every $\tau$ within the final confidence interval $[ l, (n_{11}+n_{00})/n ]$, there exists a potential table $\boldsymbol{N}$ compatible with the observed table such that $\tau(\boldsymbol{N})=\tau$ and $p_1(\boldsymbol{N})\geq \alpha$. We comment on the computational details in Section \ref{sec:more_computation_detail}, and show some numerical examples in Table \ref{table:numerical_ex_one_side}.

\begin{table}[ht]
\centering
\caption{$95\%$ one-sided confidence interval}
\label{table:numerical_ex_one_side}
\begin{tabular}{|c|c||c|c|}
\hline
$\boldsymbol{n}$ & CI & $\boldsymbol{n}$ & CI\\
\hline
$(1,1,1,13)$ & $[-1,14]$ & $(2,6,8,0)$ & $[-14,2]$\\
$(6,0,11,3)$ & $[-3, 9]$ & $(6,4,4,6)$ & $[-3,12]$\\
$(1,1,3,19)$ & $[-3,20]$ & $(8,4,5,7)$ & $[-2,15]$\\
\hline
\end{tabular}
\end{table}

\section{Proof of the theorems}
\label{sec:proof_for_theorems}

\begin{proof}[Proof of Theorem \ref{compat_st}]
Let $x_{ik}$ denote the number of units in set $\{j : y_j(1)=i, y_j(0)=k\} $ that are assigned to treatment. We have
\begin{align*} 
&x_{11}+x_{10}+x_{01}+x_{00}=  m,\quad 
x_{11} + x_{10} =  n_{11},\quad 
(N_{11}-x_{11}) + (N_{01}-x_{01}) =  n_{01},\\
&0 \leq x_{ik} \leq N_{ik} \quad (i,k=0,1). 
\end{align*}
The potential table $\boldsymbol{N}$ is compatible with the observed table $\boldsymbol{n}$ if and only if the above equations have integer solutions for $(x_{11}, x_{10}, x_{01}, x_{00})$. The above equations are equivalent to
\begin{align*}
&x_{10} =  n_{11}-x_{11},\quad 
x_{01} =  N_{11}+N_{01}-n_{01}-x_{11},\quad 
x_{00} =  x_{11}+n_{01}+n_{10}-N_{11}-N_{01},\\
&0\leq x_{ik} \leq  N_{ik}\quad  (i,k=0,1).  
\end{align*}
The integer solutions exist if and only if there exists an integer $x_{11}$ satisfying
\begin{align*} 
&0 \leq  x_{11} \leq N_{11},\quad 
0 \leq  n_{11}-x_{11} \leq N_{10},\\
&0 \leq  N_{11}+N_{01}-n_{01}-x_{11} \leq N_{01},\quad 
0 \leq  x_{11}+n_{01}+n_{10}-N_{11}-N_{01} \leq N_{00},
\end{align*}
which are equivalent to
\begin{align*}
&\max\{0, n_{11}-N_{10}, N_{11}-n_{01}, N_{11}+N_{01}-n_{10}-n_{01}\}  \\
\leq x_{11} \leq &\min\{N_{11}, n_{11}, N_{11}+N_{01}-n_{01}, N_{00}+N_{11}+N_{01}-n_{01}-n_{10}\}.
\end{align*}
Therefore, we have proved Theorem \ref{compat_st}.
\end{proof}

In order to prove the theorems, we need to introduce additional notation and lemmas. Define
\begin{align*}
\Tau &=  \{(0,1,0,-1), (-1,1,0,0), (1,0,-1,0), (0,0,-1,1) \},\\
\Tau_0 &=  \{(-1,1,0,0),  (0,0,-1,1) \} \subset \Tau.
\end{align*}

If two potential tables of sample size $n$, $\boldsymbol{N}$ and $\boldsymbol{N}'$, satisfy $\boldsymbol{N}'=\boldsymbol{N}+\boldsymbol{\Delta}$ with $\boldsymbol{\Delta} \in \Tau$, then $\tau( \boldsymbol{N}' ) = \tau(  \boldsymbol{N})  + 1/n$.

\begin{lemma}\label{lemma::diff1}
If $\boldsymbol{N}'=\boldsymbol{N}+\boldsymbol{\Delta}$ with $\boldsymbol{\Delta} \in \Tau$, then we can construct potential outcomes $\{y_j(1), y_j(0)\}_{j=1}^n$ and $\{y_j'(1), y_j'(0)\}_{j=1}^n$ such that only one unit $r$ is different, i.e., 
$(y_r(1),y_r(0))=(a_1,a_2)$, $(y'_r(1),y'_r(0))=(b_1,b_2)$, and $(y_j(1),y_j(0))=(y'_j(1),y'_j(0))$ for all $ j \neq r.$ We show the corresponding values of $\boldsymbol{\Delta}$, $(a_1,a_2)$ and $(b_1,b_2)$  in Table \ref{tb::twotables}.

\begin{table}[ht]
\caption{The difference between two potential tables}\label{tb::twotables}
\begin{center}
\begin{tabular}{|c|c|c|}
\hline
$\boldsymbol{\Delta}$ & $(a_1,a_2)$ & $(b_1,b_2)$\\
\hline
$(0,1,0,-1)$ & $(0,0)$ & $(1,0)$ \\
$(-1,1,0,0)$ & $(1,1)$ & $(1,0)$ \\
$(1,0,-1,0)$ & $(0,1)$ & $(1,1)$ \\
$(0,0,-1,1)$ & $(0,1)$ & $(0,0)$ \\
\hline
\end{tabular}
\end{center}
\end{table}
\end{lemma}

\begin{lemma}\label{lemma:fundamental_onemove}
Assume that $\boldsymbol{N}'=\boldsymbol{N}+\boldsymbol{\Delta}$ with $\boldsymbol{\Delta} \in \Tau$, and $\boldsymbol{N}$ and $\boldsymbol{N}'$ differ by only one unit $r$ as constructed in Lemma \ref{lemma::diff1}. Let $\delta$ be the random indicator for the $r$th unit being assigned to treatment. The relationship between $\hat{\tau}'  - \hat{\tau}$ and $\boldsymbol{\Delta}$ is shown in Table \ref{tb::diff-tau}.

\begin{table}[ht]
\caption{Difference between $\hat{\tau}'$ and $\hat{\tau}$}\label{tb::diff-tau}
\begin{center}
\begin{tabular}{|c|c|}
\hline
$\boldsymbol{\Delta}$ & $\hat{\tau}'  - \hat{\tau}$\\
\hline
$(0,1,0,-1)$ & $ \delta / m$ \\
$(-1,1,0,0)$ & $ (1-\delta)/ (n-m)$  \\
$(1,0,-1,0)$ & $ \delta  / m $ \\
$(0,0,-1,1)$ & $ (1-\delta)/ (n-m)$ \\
\hline
\end{tabular}
\end{center}
\end{table}
\end{lemma}

\begin{proof}[Proof of Lemma \ref{lemma:fundamental_onemove}]
Let $\omega$ be an outcome of the sample space, and treatment assignment $\boldsymbol{Z}(\omega)$ is a function of $\omega.$ For each assignment $\boldsymbol{Z}(\omega)$, let $x_{ik}(\omega)$ denote the number of units in set $\{ j: (y_j(1),y_j(0))=(i,k) \} $ that are assigned treatment, and $x'_{ik}(\omega)$ the number of units in set $\{j: (y'_j(1),y'_j(0))=(i,k) \}$ that are assigned treatment, $i,k=0,1$. Then we have
\begin{align*}
\hat{\tau}(\omega) = & \frac{1}{m}(x_{11}(\omega)+x_{10}(\omega)) - \frac{1}{n-m}\left\{ (N_{11}-x_{11}(\omega)) +(N_{01}-x_{01}(\omega))\right\},\\
\hat{\tau}'(\omega)
=&\frac{1}{m}(x'_{11}(\omega)+x'_{10}(\omega)) - \frac{1}{n-m}\left\{ (N'_{11}-x'_{11}(\omega)) +(N'_{01}-x'_{01}(\omega))\right\}.
\end{align*}

Because $N_{11}', N_{01}', x'_{ik}(\omega) $ are functions of $N_{11}, N_{01}, x_{ik} , \delta(\omega)$, as shown in Table \ref{x_quantity_N_prime}, we can immediately obtain $\hat{\tau}(\omega) - \hat{\tau}'(\omega)$ as shown in Table \ref{tb::diff-tau}.

\begin{table}[htbf]
\caption{Quantities of $\boldsymbol{N}'$ as functions of quantities of $\boldsymbol{N}$}
\label{x_quantity_N_prime}
\centering
\begin{tabular}{|c|c|c|c|c|c|c|}
\hline
$\boldsymbol{\Delta}$ & $x_{11}'(\omega)$&  $x_{10}'(\omega)$& $x'_{01}(\omega)$ &$N'_{11}$&$N'_{01}$ & $\tau(\boldsymbol{N}')$\\
\hline
$(0,1,0,-1)$ & $x_{11}(\omega)$ &  $x_{10}(\omega)+\delta(\omega)$ & $x_{01}(\omega)$ &  $N_{11}$& $N_{01}$ & $\tau(\boldsymbol{N})+1/n$\\
$(-1,1,0,0)$ & $x_{11}(\omega)-\delta(\omega)$ &  $x_{10}(\omega)+\delta(\omega)$ & $x_{01}(\omega)$ &  $N_{11}-1$& $N_{01}$ & $\tau(\boldsymbol{N})+1/n$\\
$(1,0,-1,0)$ & $x_{11}(\omega)+\delta(\omega)$ &  $x_{10}(\omega)$ & $x_{01}(\omega)-\delta(\omega)$ &  $N_{11}+1$& $N_{01}-1$ & $\tau(\boldsymbol{N})+1/n$\\
$(0,0,-1,1)$ & $x_{11}(\omega)$ &  $x_{10}(\omega)$ & $x_{01}(\omega)-\delta(\omega)$ &  $N_{11}$& $N_{01}-1$ & $\tau(\boldsymbol{N})+1/n$\\
\hline
\end{tabular}
\end{table}
\end{proof}

To prove Theorem \ref{thm::order-two-sided}, we need the following lemmas.

\begin{lemma}\label{lemma:order_p2_onemove_unbal}
If $m\leq n/2$, assume $\boldsymbol{N}'=\boldsymbol{N}+\boldsymbol{\Delta}$ with
$
\boldsymbol{\Delta}
\in \Tau_0.
$
If $\tau(\boldsymbol{N}) \leq \hat{\tau}^{\text{obs}}$ and $ \tau(\boldsymbol{N'}) \leq \hat{\tau}^{\text{obs}}$,
then $p_2(\boldsymbol{N}') \geq p_2(\boldsymbol{N})$.
\end{lemma}
\begin{proof}[Proof of Lemma \ref{lemma:order_p2_onemove_unbal}]
We consider two potential tables $\boldsymbol{N}$ and $\boldsymbol{N}'$ constructed in Lemma \ref{lemma:fundamental_onemove}. For any treatment assignment $\boldsymbol{Z}(\omega)$ and $\boldsymbol{\Delta} \in \Tau $, we have $\hat{\tau}(\omega) \leq \hat{\tau}'(\omega)\leq \hat{\tau}(\omega)+1/(n-m) \leq \hat{\tau}(\omega)+2/n$, and
$$
2\tau' - \hat{\tau}'(\omega) = 2\tau+\frac{2}{n}- \hat{\tau}'(\omega) \geq 2\tau+ \frac{2}{n}- \hat{\tau}(\omega)-\frac{2}{n} = 2\tau- \hat{\tau}(\omega).
$$
If $\tau \leq \hat{\tau}^{\text{obs}}$ and $ \tau' \leq \hat{\tau}^{\text{obs}}$, then
\begin{align*}
p_2(\boldsymbol{N}') = & P_{\boldsymbol{N}'}(|\hat{\tau}'-\tau'|\geq \hat{\tau}^{\text{obs}}-\tau') = P_{\boldsymbol{N}'}(\max\{\hat{\tau}',2\tau'-\hat{\tau}'\} \geq \hat{\tau}^{\text{obs}})\\
\geq & P_{\boldsymbol{N}}(\max\{\hat{\tau},2\tau-\hat{\tau}\} \geq \hat{\tau}^{\text{obs}})
= P_{\boldsymbol{N}}(|\hat{\tau}-\tau|\geq \hat{\tau}^{\text{obs}}-\tau) =p_2(\boldsymbol{N}).
\end{align*}
Here we change the probability measures from $P_{\boldsymbol{N}'}$ to $P_{\boldsymbol{N}}$ because of the coupling in Lemma \ref{lemma::diff1} and
Lemma \ref{lemma:fundamental_onemove}, i.e.,
\begin{align*}
P_{\boldsymbol{N}'}(\hat{\tau}' \geq & \hat{\tau}^{\text{obs}}) = \binom{n}{m}^{-1}\sum_{\omega} 1\{\max\{\hat{\tau}'(\omega),2\tau'-\hat{\tau}'(\omega)\} \geq \hat{\tau}^{\text{obs}}\} \\
\geq & \binom{n}{m}^{-1}\sum_{\omega} 1\{\max\{\hat{\tau}(\omega),2\tau-\hat{\tau}(\omega)\} \geq \hat{\tau}^{\text{obs}} \} = P_{\boldsymbol{N}}(\hat{\tau} \geq \hat{\tau}^{\text{obs}}).
\end{align*}
\end{proof}

\begin{lemma}\label{lemma:order_p2_unbal}
If $m\leq n/2$, assume $\boldsymbol{N}$ and $\boldsymbol{N}'$ are two potential tables. If $\tau(\boldsymbol{N}) \leq \hat{\tau}^{\text{obs}}, \tau(\boldsymbol{N'}) \leq \hat{\tau}^{\text{obs}}$, and
$
N_{11}'=N_{11}, N'_{10} = N_{10}, N'_{01} \leq N_{01} ,
$
then $p_2(\boldsymbol{N}') \geq p_2(\boldsymbol{N})$.

\end{lemma}

\begin{proof}[Proof of Lemma \ref{lemma:order_p2_unbal}]
Because $\tau(\boldsymbol{N}) \leq \hat{\tau}^{\text{obs}}, \tau(\boldsymbol{N'}) \leq \hat{\tau}^{\text{obs}}$, and
\begin{align*}
 \boldsymbol{N}+(N_{01}-N_{01}') \cdot (0,0,-1,1) 
= (N'_{11}, N'_{10}, N'_{01}, N_{00}') = \boldsymbol{N}',
\end{align*}
we know $p_2(\boldsymbol{N}') \geq p_2(\boldsymbol{N})$ according to Lemma \ref{lemma:order_p2_onemove_unbal}.
\end{proof}

\begin{lemma}\label{lemma:order_p2_onemove}
In balanced experiments, assume $\boldsymbol{N}'=\boldsymbol{N}+\boldsymbol{\Delta}$ with
$
\boldsymbol{\Delta}
\in \Tau.
$
If $\tau(\boldsymbol{N}) \leq \hat{\tau}^{\text{obs}}$ and $ \tau(\boldsymbol{N'}) \leq \hat{\tau}^{\text{obs}}$,
then $p_2(\boldsymbol{N}') \geq p_2(\boldsymbol{N})$.
\end{lemma}

\begin{proof}[Proof of Lemma \ref{lemma:order_p2_onemove}]
We consider two potential tables $\boldsymbol{N}$ and $\boldsymbol{N}'$ constructed in Lemma \ref{lemma:fundamental_onemove}. For any treatment assignment $\boldsymbol{Z}(\omega)$ and $\boldsymbol{\Delta} \in \Tau $, we have $\hat{\tau}(\omega) \leq \hat{\tau}'(\omega)\leq \hat{\tau}(\omega)+1/m$, and
$$
2\tau' - \hat{\tau}'(\omega) = 2\tau+\frac{2}{n}- \hat{\tau}'(\omega) \geq 2\tau+ \frac{2}{n}- \hat{\tau}(\omega)-\frac{1}{m} = 2\tau- \hat{\tau}(\omega).
$$
If $\tau \leq \hat{\tau}^{\text{obs}}$ and $ \tau' \leq \hat{\tau}^{\text{obs}}$, then
\begin{align*}
p_2(\boldsymbol{N}') = & P_{\boldsymbol{N}'}(|\hat{\tau}'-\tau'|\geq \hat{\tau}^{\text{obs}}-\tau') = P_{\boldsymbol{N}'}(\max\{\hat{\tau}',2\tau'-\hat{\tau}'\} \geq \hat{\tau}^{\text{obs}})\\
\geq & P_{\boldsymbol{N}}(\max\{\hat{\tau},2\tau-\hat{\tau}\} \geq \hat{\tau}^{\text{obs}})
= P_{\boldsymbol{N}}(|\hat{\tau}-\tau|\geq \hat{\tau}^{\text{obs}}-\tau) =p_2(\boldsymbol{N}).
\end{align*}
\end{proof}

\begin{lemma}\label{lemma:order_p2}
In balanced experiments, assume $\boldsymbol{N}$ and $\boldsymbol{N}'$ are two potential tables. If $\tau(\boldsymbol{N}) \leq \hat{\tau}^{\text{obs}}, \tau(\boldsymbol{N'}) \leq \hat{\tau}^{\text{obs}}$, and
$
N_{11}'=N_{11}, N'_{10} \geq N_{10}, N'_{01} \leq N_{01} ,
$
then $p_2(\boldsymbol{N}') \geq p_2(\boldsymbol{N})$.
\end{lemma}

\begin{proof}[Proof of Lemma \ref{lemma:order_p2}]
Because $\tau(\boldsymbol{N}) \leq \hat{\tau}^{\text{obs}}, \tau(\boldsymbol{N'}) \leq \hat{\tau}^{\text{obs}}$, and
\begin{align*}
 \boldsymbol{N}+(N_{01}-N_{01}') \cdot (0,0,-1,1) + (N_{10}'-N_{10})\cdot (0,1,0,-1) 
= (N'_{11}, N'_{10}, N'_{01}, N_{00}') = \boldsymbol{N}',
\end{align*}
we know $p_2(\boldsymbol{N}') \geq p_2(\boldsymbol{N})$ according to Lemma \ref{lemma:order_p2_onemove}.
\end{proof}

\begin{proof}[{\bf Proof of Theorem \ref{thm::order-two-sided}}]
We first prove (1). We show $\underline{N}_{10}(N_{11}, N_{01}) \leq \underline{N}_{10}(N_{11}, N_{01}')$, for any $0 \leq N_{11} \leq n$ and $ 0   \leq N_{01}\leq N_{01}' \leq n-N_{11}$. If $\underline{N}_{10}(N_{11}, N_{01}') \geq \lfloor N_{01} + n \hat{\tau}^{\obs} \rfloor+1$, then the conclusion holds trivially. If $\underline{N}_{10}(N_{11}, N_{01}') \leq \lfloor N_{01} + n \hat{\tau}^{\obs} \rfloor  \leq \lfloor N_{01}' + n \hat{\tau}^{\obs} \rfloor$, then
$$
\boldsymbol{N}_1 = (N_{11}, \underline{N}_{10}(N_{11}, N_{01}'), N_{01}', n - N_{11} -\underline{N}_{10}(N_{11}, N_{01}')-N_{01}' )
$$ 
is a potential table satisfying $\tau(\boldsymbol{N}_1) \leq \hat{\tau}^{\text{obs}}$ and $p_2(\boldsymbol{N}_1)\geq \alpha$. Because 
$$
N_{11} + \underline{N}_{10}(N_{11}, N_{01}') + N_{01} \leq N_{11} + \underline{N}_{10}(N_{11}, N_{01}') + N_{01}' \leq n,
$$
we know that
$$
\boldsymbol{N}_2=(N_{11},  \underline{N}_{10}(N_{11}, N_{01}'), N_{01}, n- N_{11} - \underline{N}_{10}(N_{11}, N_{01}') - N_{01})
$$
is a potential table which
satisfies
$$
\tau(\boldsymbol{N}_2) = (\underline{N}_{10}(N_{11}, N_{01}') - N_{01})/n \leq  (\lfloor N_{01} + n \hat{\tau}^{\obs} \rfloor  -N_{01})/n \leq \hat{\tau}^{\obs}.
$$ 
According to Lemma \ref{lemma:order_p2_unbal}, $p_2(\boldsymbol{N}_2)\geq p_2(\boldsymbol{N}_1) \geq \alpha$, implying
$
\underline{N}_{10}(N_{11}, N_{01}) \leq \underline{N}_{10}(N_{11}, N_{01}').
$

We then prove (2). 
For a potential table $\boldsymbol{N}$ with $ \tau(\boldsymbol{N}) \leq \hat{\tau}^{\text{\obs}}$, if $N_{10}<\underline{N}_{10}(N_{11},N_{01})$, then $p_1(\boldsymbol{N})<\alpha$ from the definition of $\underline{N}_{10}(N_{11},N_{01})$. 
Otherwise, $\underline{N}_{10}(N_{11},N_{01}) \leq N_{10} \leq \lfloor N_{01}+n\hat{\tau}^\obs \rfloor$, which implies that
$$
\boldsymbol{N}_3 = (N_{11}, \underline{N}_{10}(N_{11}, N_{01}), N_{01}, n-N_{11}-\underline{N}_{10}(N_{11}, N_{01})-N_{01})
$$
is a potential table satisfying $\tau(\boldsymbol{N}_3 ) \leq \hat{\tau}^{\text{obs}}$ and $p_2(\boldsymbol{N}_3) \geq \alpha$. According to Lemma \ref{lemma:order_p2}, $p_2(\boldsymbol{N}) \geq p_2(\boldsymbol{N}_3) \geq \alpha$.

\end{proof}

To prove Theorem \ref{thm::order-one-sided}, we need the following lemmas.

\begin{lemma}\label{lemma:order_p1_onemove}
If $\boldsymbol{N}'=\boldsymbol{N}+\boldsymbol{\Delta}$ with
$
\boldsymbol{\Delta}
\in \Tau,
$
then $p_1(\boldsymbol{N}') \geq p_1(\boldsymbol{N})$.
\end{lemma}

\begin{proof}[Proof of Lemma \ref{lemma:order_p1_onemove}]
We consider the potential tables constructed in Lemma \ref{lemma::diff1}. For any treatment assignment $\boldsymbol{Z}(\omega)$, we have $\hat{\tau}'(\omega)\geq \hat{\tau}(\omega)$, and therefore $p_1(\boldsymbol{N}') = P_{\boldsymbol{N}'}(\hat{\tau}' \geq \hat{\tau}^{\text{obs}}) \geq P_{\boldsymbol{N}}(\hat{\tau} \geq \hat{\tau}^{\text{obs}}) = p_1(\boldsymbol{N})$.
\end{proof}

\begin{lemma}\label{lemma:order_p1}
If $\boldsymbol{N}'$ and $\boldsymbol{N}$ satisfy
$
N_{11}'=N_{11}, N'_{10} \geq N_{10}
$
and 
$
N'_{01} \leq N_{01},
$
then $p_1(\boldsymbol{N}') \geq p_1(\boldsymbol{N})$.
\end{lemma}

\begin{proof}[Proof of Lemma \ref{lemma:order_p1}]
Because
\begin{align*}
 \boldsymbol{N}+(N_{01}-N_{01}') \cdot (0,0,-1,1) + (N_{10}'-N_{10})\cdot (0,1,0,-1)
=  (N'_{11}, N'_{10}, N'_{01}, N_{00}') = \boldsymbol{N}',
\end{align*} 
we have $p_1(\boldsymbol{N}') \geq p_1(\boldsymbol{N})$ by repeatedly applying Lemma \ref{lemma:order_p1_onemove}.
\end{proof}

\begin{proof}[{\bf Proof of Theorem \ref{thm::order-one-sided}}]
(1) If $\underline{N}_{10}(N_{11},N'_{01}) = n+1$, the conclusion holds trivially. If $\underline{N}_{10}(N_{11},N'_{01}) \leq n$, then
$$
\boldsymbol{N}_1= (N_{11}, \underline{N}_{10}(N_{11},N'_{01}), N_{01}', n- N_{11}-  \underline{N}_{10}(N_{11},N'_{01})- N_{01}')
$$
is a  potential table and satisfies $p_1(\boldsymbol{N}_1)\geq \alpha$. Because $N_{01}\leq N_{01}'$, 
$$
\boldsymbol{N}_2=(N_{11}, \underline{N}_{10}(N_{11},N'_{01}), N_{01}, n-N_{11}- \underline{N}_{10}(N_{11},N'_{01})-N_{01})
$$
is also a potential table. Lemma \ref{lemma:order_p1} further implies $
p_1(\boldsymbol{N}_2) \geq p_1(\boldsymbol{N}_1) \geq \alpha.
$
Therefore $\underline{N}_{10}(N_{11},N_{01}) \leq \underline{N}_{10}(N_{11},N'_{01})$.

(2) For a potential table $\boldsymbol{N}$, if $N_{10}<\underline{N}_{10}(N_{11},N_{01})$, then $p_1(\boldsymbol{N})<\alpha$ by the definition of $\underline{N}_{10}(N_{11},N_{01})$. Otherwise, $\underline{N}_{10}(N_{11},N_{01}) \leq N_{10} <n+1$, implying that
$$
\boldsymbol{N}_3 = (N_{11}, \underline{N}_{10}(N_{11}, N_{01}), N_{01}, n-N_{11}-\underline{N}_{10}(N_{11}, N_{01})-N_{01})
$$
is a potential table and satisfies $p_1(\boldsymbol{N}_3) \geq \alpha$. According to Lemma \ref{lemma:order_p1},  
$
p_1(\boldsymbol{N}) \geq p_1(\boldsymbol{N}_3) \geq \alpha.
$

\end{proof}

\section{More computational details}
\label{sec:more_computation_detail}

\begin{theorem}
\label{thm::computation-two-sided}
The procedure for the two-sided confidence interval in Section \ref{subsec:two_side} requires at most $O(n^2)$ randomization tests.
\end{theorem}

\begin{proof}[Proof of Theorem \ref{thm::computation-two-sided}]
For any fixed $N_{11}$, it must be true that $\underline{N}_{10}(N_{11},N_{01}) = \lfloor N_{01} + n \hat{\tau}^{\obs} \rfloor+1$ for $0\leq N_{01} <  0 \vee \lceil -n \hat{\tau}^{\text{obs}}\rceil$ without doing any randomization tests, where $\lceil x \rceil$ denotes the smallest integer greater than or equal to $x$.
 Hence,
in order to get $\underline{N}_{10}(N_{11}, N_{01})$ for all $ 0 \leq N_{01}\leq n-N_{11}$, the number of randomization tests needed is less than or equal to 
$$ n \wedge \underline{N}_{10}(N_{11},  0 \vee \lceil -n \hat{\tau}^{\text{obs}}\rceil) + 1 + \sum_{k= 0 \vee \lceil -n \hat{\tau}^{\text{obs}}\rceil+1}^{n-N_{11}}\{ n \wedge \underline{N}_{10}(N_{11}, k)- n \wedge \underline{N}_{10}(N_{11}, k-1)+1\} \leq 2n+1,$$
where $n \wedge \underline{N}_{10}(N_{11},  0 \vee \lceil -n \hat{\tau}^{\text{obs}}\rceil) + 1$ bounds the number of randomization tests needed for $N_{01}=0 \vee \lceil -n \hat{\tau}^{\text{obs}}\rceil$, and $n \wedge \underline{N}_{10}(N_{11}, k)- n \wedge \underline{N}_{10}(N_{11}, k-1)+1$ bounds the number of randomization tests needed for $N_{01}=k\ (0 \vee \lceil -n \hat{\tau}^{\text{obs}}\rceil+1\leq k\leq n-N_{11})$.

Because the number of possible values of $N_{11}$ is less than $(n+1)$, the total number of randomization tests needed for calculating the lower confidence limit is less than or equal to $(2n+1)(n+1)=O(n^2)$. The computation for the upper limit of $\tau$ is the same as the lower limit by switching the labels of $Y$. Therefore, the total number of randomization tests needed is at most $O(n^2)$.
\end{proof}

\begin{theorem}
\label{thm::computation-one-sided}
The procedure for the one-sided confidence interval in Section \ref{subsec:one_side} requires at most $O(n^2)$ randomization tests.
\end{theorem}

\begin{proof}[Proof of Theorem \ref{thm::computation-one-sided}]
For any fixed $N_{11}$, in order to get $\underline{N}_{10}(N_{11}, N_{01})$ for all $0\leq N_{01}\leq n-N_{11}$, the number of randomization test needed is less than or equal to
$$
n \wedge \underline{N}_{10}(N_{11}, 0) + 1 + \sum_{k=1}^{n-N_{11}}\{n \wedge  \underline{N}_{10}(N_{11}, k)- n \wedge \underline{N}_{10}(N_{11}, k-1)+1\} \leq 2n+1,
$$
where $n \wedge \underline{N}_{10}(N_{11}, 0) + 1$ bounds the number of randomization tests needed for $N_{01}=0$, and $n \wedge \underline{N}_{10}(N_{11}, k)- n \wedge \underline{N}_{10}(N_{11}, k-1)+1$ bounds the number of randomization tests needed for $N_{01}=k\  (1\leq k\leq n-N_{11})$.

Because the number of possible values of $N_{11}$ is less than $(n+1)$, the total number of randomization tests needed is less than or equal to $(2n+1)(n+1)=O(n^2)$.
\end{proof}

Mathematically, by inverting a series of randomization tests, we obtain confidence sets for $\tau.$ These confidence sets in Sections \ref{subsec:two_side} and \ref{subsec:one_side} may not be intervals. The final theorems rule out this possibility, and confirm that these confidence sets are indeed confidence intervals. In order to prove the final two theorems, we need to introduce the following lemma.

\begin{lemma}
\label{lemma:exist_compat_potential_table}
For any potential table $\boldsymbol{N}$ compatible with the observed table $\boldsymbol{n}$,
if $\tau(\boldsymbol{N})<(n_{11}+n_{00})/n$, then there exist a potential table $\boldsymbol{N} ' $ such that $\boldsymbol{N} ' $ is compatible with the observed table and $\boldsymbol{N} '  = \boldsymbol{N} + \boldsymbol{\Delta}$ with $\boldsymbol{\Delta} \in \Tau$.
\end{lemma}
\begin{proof}[Proof of Lemma \ref{lemma:exist_compat_potential_table}]
Because $\boldsymbol{N}$ is compatible with the observed table $\boldsymbol{n}$, there exist potential outcomes $\{y_j(1), y_j(0)\}_{j=1}^n$, summarized by $\boldsymbol{N}$, that give the observed table $\boldsymbol{n}$ under the treatment assignment $\boldsymbol{Z}$.
We construct potential outcomes $\{y_j'(1), y_j'(0)\}_{j=1}^n$ different from $\{y_j(1), y_j(0)\}_{j=1}^n$ by only one unit $r$, i.e. $Z_r = z, (y_r(1), y_r(0))=(a_1, a_2), (y'_r(1), y'_r(0))=(b_1, b_2)$, and $(y_j(1),y_j(0))=(y'_j(1),y'_j(0))$ for all $j\neq r$. We show the corresponding values of $z, (a_1, a_2)$ and $(b_1,b_2)$ in Table \ref{table:construct_compat_N_1}, where $\boldsymbol{N} ' $ denotes the potential table summarizing $\{y_j'(1), y_j'(0)\}_{j=1}^n$, and $\boldsymbol{\Delta}=\boldsymbol{N}_1-\boldsymbol{N}$.

\begin{table}[ht]
\caption{Constructing potential table $\boldsymbol{N} ' $}
\label{table:construct_compat_N_1}
\begin{center}
\begin{tabular}{|c|c|c|c|}
\hline
$z$ & $(a_1,a_2)$ & $(b_1, b_2)$ & $\boldsymbol{N} ' -\boldsymbol{N}=\boldsymbol{\Delta}$\\
\hline
$0$ & $(0,0)$ & $(1,0)$ & $(0,1,0,-1)$\\
$1$& $(1,1)$ & $(1,0)$ & $(-1,1,0,0)$\\
$0$& $(0,1)$ & $(1,1)$ & $(1,0,-1,0)$\\
$1$ & $(0,1)$ & $(0,0)$ & $(0,0,-1,1)$\\
\hline
\end{tabular}
\end{center}
\end{table}

Because potential outcomes $\{y_j'(1), y_j'(0)\}_{j=1}^n$ give the same observed table $\boldsymbol{n}$ under the treatment assignment $\boldsymbol{Z}$, $\boldsymbol{N} ' $ is compatible with the observed table. We need only to show that unit $r$ exists if $\tau(\boldsymbol{N})<(n_{11}+n_{00})/n$. If such unit $r$ does not exist, then the following must be true:
\begin{align}\label{eq:conditon_not_exist}
N_{00}-x_{00}=0,  \quad x_{11}=0, \quad N_{01}-x_{01}=0, \quad x_{01}=0,
\end{align}
recalling that $x_{ik}$ denotes the number of units in set $\{j : y_j(1)=i, y_j(0)=k\} $ that are assigned to treatment under the treatment assignment $\boldsymbol{Z}$. Formula (\ref{eq:conditon_not_exist}) implies $\boldsymbol{N}=(n_{01}, n_{11}+n_{00}, 0, n_{10})$ and $\tau(\boldsymbol{N})=(n_{11}+n_{00})/n$, which contradicts $\tau(\boldsymbol{N})<(n_{11}+n_{00})/n$. Therefore, (\ref{eq:conditon_not_exist}) cannot hold and the unit $r$ must exist, and Lemma \ref{lemma:exist_compat_potential_table} holds.
\end{proof}

\begin{theorem}
\label{thm::form_two_sided}
In balanced experiments,
the final two-sided confidence set in Section \ref{subsec:two_side}  must have the form $[ l, u]$ for some values $l$ and $u$, in the sense that for every possible $\tau \in [l,u]$, there exists a potential table $\boldsymbol{N}$ compatible with observed table that satisfies $\tau(\boldsymbol{N})=\tau$ and $p_2(\boldsymbol{N})\geq \alpha$.
\end{theorem}
\begin{proof}[Proof of Theorem \ref{thm::form_two_sided}]
For any $\tau < \hat{\tau}^\obs$, if there exists a potential table $\boldsymbol{N}$ satisfying $\tau(\boldsymbol{N})=\tau$ and $p_2(\boldsymbol{N})\geq \alpha$, then according to Lemma \ref{lemma:exist_compat_potential_table} and Lemma \ref{lemma:order_p2_onemove}, there exists a potential table $\boldsymbol{N}'$ satisfying $\tau(\boldsymbol{N}')=\tau+1/n\leq \hat{\tau}^\obs$ and $p_2(\boldsymbol{N}')\geq \alpha$. Similarly, by changing the labels of $Y$, we know that for any $\tau > \hat{\tau}^\obs$, if there exists a potential table $\boldsymbol{N}$ satisfying $\tau(\boldsymbol{N})=\tau$ and $p_2(\boldsymbol{N})\geq \alpha$, then there exists a potential table $\boldsymbol{N}''$ satisfying $\tau(\boldsymbol{N}'')=\tau-1/n$ and $p_2(\boldsymbol{N}'')\geq \alpha$. Therefore, Theorem \ref{thm::form_two_sided} holds.
\end{proof}

\begin{theorem}
\label{thm::form_one_sided}
The final confidence set in Section \ref{subsec:one_side}
must have the form $[ l, (n_{11}+n_{00})/n ]$ for some value $l$, in the sense that for every possible $\tau\geq l$, there exists a potential table $\boldsymbol{N}$ compatible with observed table that satisfies $\tau(\boldsymbol{N})=\tau$ and $p_1(\boldsymbol{N})\geq \alpha$.
\end{theorem}
\begin{proof}[Proof of Theorem \ref{thm::form_one_sided}]
For any possible $\tau<(n_{11}+n_{00})/n$, if there exists a potential table $\boldsymbol{N}$ satisfying $\tau(\boldsymbol{N})=\tau$ and $p_1(\boldsymbol{N})\geq \alpha$, then according to Lemma \ref{lemma:exist_compat_potential_table} and Lemma \ref{lemma:order_p1_onemove}, there exists a potential table $\boldsymbol{N}'$ satisfying $\tau(\boldsymbol{N}')=\tau+1/n$ and $p_1(\boldsymbol{N}')\geq \alpha$. Therefore, Theorem \ref{thm::form_one_sided} holds.
\end{proof}

\end{document}